\documentclass[conference]{IEEEtran}
\IEEEoverridecommandlockouts

\usepackage{algorithm}
\usepackage{algorithmicx}
\usepackage{algpseudocode}
\usepackage{amsmath}
\usepackage{indentfirst}
\floatname{algorithm}{Algorithm}

\usepackage{cite}
\usepackage{amsmath,amssymb,amsfonts}

\usepackage{graphicx}
\usepackage{textcomp}
\usepackage{xcolor}
\usepackage{epstopdf}
\usepackage{subfigure}
\usepackage{booktabs}
\usepackage{amsthm}
\usepackage{enumerate}

\usepackage{xcolor}
\algnewcommand\BlueKeyword{\textcolor{blue}}

\newtheorem{definition}{Definition}

\newtheorem{theorem}{Theorem}

\def\BibTeX{{\rm B\kern-.05em{\sc i\kern-.025em b}\kern-.08em
    T\kern-.1667em\lower.7ex\hbox{E}\kern-.125emX}}
\begin{document}
\title{ Blockchain-assisted Twin Migration for Vehicular Metaverses: A Game Theory Approach}

\author{Yue Zhong, Jinbo Wen, Junhong Zhang, Jiawen Kang*,  Yuna Jiang, Yang Zhang, Yanyu Cheng, Yongju Tong
\thanks{
Y. Zhong, J. Zhang, J. Kang, and Y. Tong are with the Guangdong University of Technology, China (e-mail: 3220001516@mail2.gdut.edu.cn; junhong1013@163.com; kavinkang@gdut.edu.cn; 3221001073@mail2.gdut. edu.cn). J. Wen and Y. Zhang are with the Nanjing University of Aeronautics and Astronautics, China (e-mail: jinbo1608@163.com; yangzhang@nuaa.edu.cn). Y. Jiang is with the Huazhong University of Science and Technology, China (e-mail: yunajiang@hust.edu.cn). Y. Cheng is with the Nanyang Technological University, Singapore (e-mail: yanyu.cheng@ntu.edu.sg).

The work was presented in part at the 5th International Conference on Electronics and Communication, Network and Computer Technology
(\textit{*Corresponding author: Jiawen Kang}).
}
}

\maketitle

\begin{abstract}
As the fusion of automotive industry and metaverse, vehicular metaverses establish a bridge between the physical space and virtual space, providing intelligent transportation services through the integration of various technologies, such as extended reality and real-time rendering technologies, to offer immersive metaverse services for Vehicular Metaverse Users (VMUs). In vehicular metaverses, VMUs update vehicle twins (VTs) deployed in RoadSide Units (RSUs) to obtain metaverse services. However, due to the mobility of vehicles and the limited service coverage of RSUs, VT migration is necessary to ensure continuous immersive experiences for VMUs. This process requires RSUs to contribute resources for enabling efficient migration, which leads to a resource trading problem between RSUs and VMUs. Moreover, a single RSU cannot support large-scale VT migration. To this end, we propose a blockchain-assisted game approach framework for reliable VT migration in vehicular metaverses. Based on the subject logic model, we first calculate the reputation values of RSUs considering the freshness of interaction between RSUs and VMUs. Then, a coalition game based on the reputation values of RSUs is formulated, and RSU coalitions are formed to jointly provide bandwidth resources for reliable and large-scale VT migration. Subsequently, the RSU coalition with the highest utility is selected. Finally, to incentivize VMUs to participate in VT migration, we propose a Stackelberg model between the selected coalition and VMUs. Numerical results demonstrate the reliability and effectiveness of the proposed schemes.
\end{abstract}


\section{Introduction}
The metaverse is a stereoscopic virtual space that exists parallel to the physical space and has recently experienced significant advancements through cutting-edge technologies, such as Artificial Intelligence (AI), eXtended Reality (XR), and blockchain. Vehicular metaverse is defined as a future continuum between automotive industry and metaverse \cite{https://doi.org/10.48550/arxiv.2210.15109}. Vehicle Twins (VTs) that act as a critical component in the vehicular metaverse, are highly accurate and large-scale digital replicas that cover the entire life cycle of vehicles and Vehicular Metaverse Users (VMUs)\cite{Jinbo}. VMUs that consist of drivers and passengers, can access the vehicular metaverse through VTs to enjoy immersive virtual experiences, e.g., AR navigation, virtual games and virtual traveling \cite{10040983}. Virtual traveling specifically involves utilizing interactive technologies like Virtual Reality (VR) and Augmented Reality (AR) to fully immerse individuals in the metaverse's virtual world, allowing for a realistic travel experience. One application of VTs for smart driving is to forecast the collision risks as warning and safety instructions for VMUs \cite{yu2022bi}. To ensure real-time physical-virtual synchronization, vehicles and VMUs continuously update their VTs in virtual spaces by obtaining sensing data from surrounding environments through the use of smart sensors, such as real-time vehicle status and passenger bio-data \cite{xu2022epvisa}.

As computational requirements for building VTs and metaverse services may be unbearable for resource-limited vehicles\cite{9880566}, vehicles offload computation-intensive tasks to nearby edge servers in RoadSide Units (RSUs) that possess sufficient resources \cite{8736823}, such as bandwidth and computing resources\cite{yu2022bi,Jinbo}, and multiple VTs can be deployed in the RSU simultaneously. However, because of the limited service coverage of RSUs and the mobility of vehicles \cite{9124705}, a single RSU cannot continuously provide metaverse services for VMUs, requiring each VT to migrate from the current RSU to another to ensure a seamless immersive experience for VMUs. Considering that the number of VTs will increase sharply with the advent of metaverses\cite{9880566}, there are some challenges for the future development of vehicular metaverses: \textbf{C1)} Some RSUs may misbehave to decrease VT migration efficiency. \textbf{C2)} A single RSU cannot provide sufficient bandwidth resources for VT migration simultaneously. \textbf{C3)} The VMUs may be reluctant to participate in VT migration without a reasonable incentive. Some efforts have been conducted for resource optimization in vehicular metaverses\cite{xu2023generative,9880566}, but they ignore the VT migration problem due to the mobility of vehicles.

To address the above challenges, we first calculate the reputation values of RSUs based on the subjective logic model in this paper. Since blockchain is a distributed technology and can effectively prevent data tampering\cite{wen2022optimal}, we propose a blockchain system to achieve distributed secure reputation management in vehicular metaverses. Then, we formulate a coalition game among RSUs based on the reputation values to select the RSU coalition with the highest utility for reliable and large-scale VT migration. Finally, a Stackelberg model is proposed to incentivize VMUs for VT migration. The main contributions are summarized as follows:

\begin{itemize}
    \item We calculate the reputation values of RSUs based on the freshness of interaction between RSUs and VMUs and propose a blockchain-assisted reputation rating system in vehicular metaverses, where RSUs acting as miners are divided into different levels according to their reputation values for lightweight consensus. (For \textbf{C1})
    \item  We formulate a coalition game based on the reputation values of RSUs for reliable and large-scale VT migration. The reputation values are utilized to evaluate the reliability of RSUs, and a blockchain-assisted reputation rating system is proposed to manage the security of reputation. In the coalition game, RSUs form coalitions to increase their profits. (For \textbf{C2})
    \item To incentivize VMUs for VT migration, we propose a Stackelberg model between the RSU coalition with the highest utility and VMUs, where the RSU coalition acting as the leader determines the bandwidth pricing strategy and VMUs acting as followers determine the bandwidth demand strategies based on the strategy of the RSU coalition as responses to the RSU coalition. (For \textbf{C3})
\end{itemize}

The rest of the paper is described as follows. Section \ref{Related} presents the related works. Section \ref{System} demonstrates the blockchain-assisted game approach framework for VT migration in vehicular metaverses. Section \ref{Coalition} introduces the coalition game-based RSU selection for vehicular metaverses. Section \ref{Stackelberg} introduces the single-leader and multi-follower Stackelberg model between the selected RSU coalition and VMUs. Numerical results are shown in Section \ref{Results}. Section \ref{Conclusion} concludes the paper and elaborates the future work.

\section{Related Works}\label{Related}
Metaverse was first introduced in the fiction named \textit{Snow Crash} in 1992 \cite{stephenson2003snow}. In \cite{681903}, the authors studied the avatars compared to the real-time virtual human research state. Virtual space is a parallel space with physical space, and humans have begun to migrate to virtual spaces on a large scale. With the development of cutting-edge technologies, metaverse has aroused widespread attention. In \cite{encyclopedia2010031, 9667507}, the authors gave a survey and detailed introduction to metaverse, including the technologies, development, applications, and open challenges of metaverse. The vehicular metaverse can be defined as the immersive integration of vehicular communications that merge virtual spaces and real data to create emerging vehicular services for VMUs \cite{luong2022edge}, which has attracted widespread attention from scholars and the automotive industry. The authors in \cite{https://doi.org/10.48550/arxiv.2210.15109    } proposed a new term named \textit{Vetaverse}, which is defined as the future continuum between vehicular industries and metaverse.

The academic discussion on metaverse service optimization focused on two aspects: resource allocation optimization and Quality of Service (QoS) optimization. For resource allocation optimization, the authors in \cite{https://doi.org/10.48550/arxiv.2212.01325  } proposed a resource allocation framework for augmented reality-empowered vehicular metaverses to improve the utility of the system, which is considered from the perspective of resource optimization for metaverse services. The authors in \cite{9838492} proposed a stochastic optimal resource allocation scheme based on random integer programming to minimize the cost of the virtual service provider. For QoS optimization, the authors in \cite{10040983} reconsidered QoS and proposed a framework that simultaneously considered the metaverse system design, the utility of consumers, and the profitability of the Metaverse Service Providers (MSPs). The authors in \cite{9973495} proposed distributed and centralized approaches to study the joint optimization problem of user association and resource pricing for metaverses. Although a lot of work has been done to study the optimization of metaverse services, most of the existing works do not consider both resource optimization and QoS optimization in vehicular metaverses nor do they consider the VT migration problem because of the mobility of vehicles. 

With the exponential growth in data volume and value, the evolving metaverse faces service security requirements and challenges \cite{9347706,9184279,9415746}. Blockchain technologies can be utilized to meet the trusted construction, continuous data interaction, and computational needs of the metaverse \cite{9953075}. Blockchain securely stores and shares data through a decentralized system that uses cryptography to ensure security. Transactions are verified through proof of work, which makes the whole process transparent and secure without the need for a central authority \cite{huynh2023blockchain}. In \cite{gadekallu2022blockchain}, the authors discussed how to protect digital contents and data of metaverse users by using blockchain technologies that have the features of decentralization, immutability, and transparency. In \cite{9889723}, the authors proposed a system model that can transparently manage user-identifiable data in the metaverse by using blockchain technologies. In \cite{9860983}, the authors proposed MetaChain, a novel blockchain-based framework that allows MSPs to allocate their resources to meet the needs of metaverse applications and metaverse users efficiently. Many works have been conducted to examine the use of blockchain for user data preservation and resource management in the metaverse. However, the existing works have not yet devised a comprehensive security scheme based on blockchain technology to adequately safeguard the service security of VT migration.

\begin{figure*}[t]
\centering
\includegraphics[width=1\textwidth]{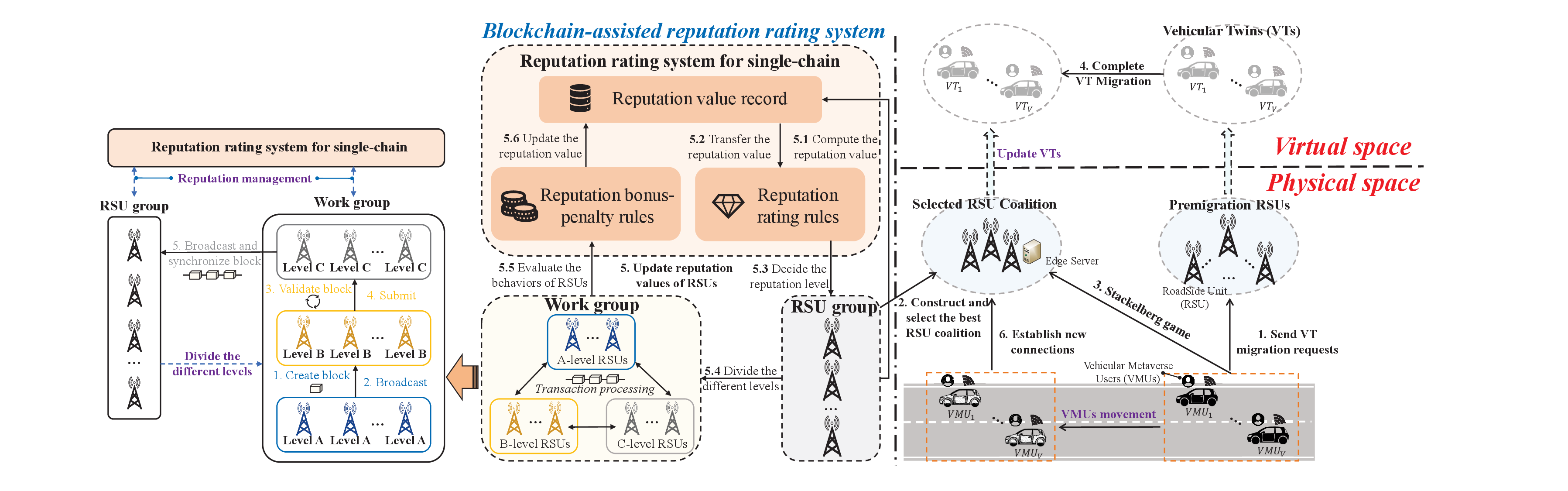}
\caption{A blockchain-assisted game approach framework for reliable VT migration in vehicular metaverses.}
\label{interaction}
\end{figure*}

\section{System Model}\label{System}
Vehicular metaverse mainly includes the physical space, the virtual space, and the interaction layer \cite{9880566,xu2022epvisa}. In the physical space, VMUs are drivers and passengers within vehicles, they can enjoy metaverse services with the application of XR technologies\cite{https://doi.org/10.48550/arxiv.2210.15109}. For example, VMUs can see virtual scenes of the front windshield and side windows through head-mounted displays\cite{9880566}. To ensure seamless immersion and interactions for VMUs, resource-limited vehicles offload the large-scale rendering tasks of updating VTs to the nearby edge servers in RSUs\cite{yu2022bi,Jinbo}, and VTs have to be correspondingly migrated from the current RSUs to other RSUs due to the limited RSU coverage and the mobility of vehicles\cite{Jinbo}. To achieve reliable and efficient VT migration, VMUs purchase sufficient bandwidth resources from well-behaved RSUs\cite{9973495}. Especially, we mathematically calculate the reputations of RSUs to quantify their reliability of RSUs. Then, we propose a blockchain-assisted reputation rating system to manage the reputation values securely. Figure \ref{interaction} shows the blockchain-assisted game approach framework for reliable VT migration in vehicular metaverses. Here, we provide further elaboration on the components of this framework as follows:

\textit{Step 1 (Send VT migration requests): }When VMUs travel on the road, the RSU cannot provide continuous metaverse services for VMUs due to the limited service coverage\cite{yu2022bi,Jinbo}. Therefore, to ensure seamless immersive experiences for VMUs, VTs should be migrated from the current RSUs that they are deployed in (i.e., premigration RSUs) to other RSUs\cite{Jinbo}. Before VT migration begins, VMUs send VT migration requests to the premigration RSUs.

\textit{Step 2 (Construct and select the best RSU coalition): }Upon receiving the VT migration requests, the premigration RSUs broadcast the requests to their surrounding RSUs. Then, the surrounding RSUs first form RSU coalitions based on the reputation values of RSUs. The reputation values are calculated based on the subjective logic model\cite{8832210,9880566}, which are recorded and managed securely on the blockchain. Finally, to ensure reliable and efficient VT migration, the RSU coalition with the highest utility is selected to provide bandwidth resources for VMUs.

\textit{Step 3 (Stackelberg game between the selected RSU coalition and VMUs):} In the VT migration, the selected RSU coalition is the sole bandwidth resource holder and VMUs purchase bandwidth resources from the RSU coalition to migrate VTs from the premigration RSUs to the selected RSU coalition. To maximize the profit of the RSU coalition and maintain its monopoly power, a single-leader and multi-follower Stackelberg model is proposed, which consists of two stages, as shown in Fig. \ref{Stac}. In the first stage, the RSU coalition acting as the leader determines the selling price of unit bandwidth. In the second stage, the VMUs acting as followers determine the amount of bandwidth to purchase based on the pricing strategy of the RSU coalition. This can also be regarded as a resource pricing optimization problem. From the perspective of the selected RSU, it can optimize resource allocation, and from the perspective of VMU, it can optimize service quality.

\textit{Step 4 (Complete VT migration): }Based on the optimal selling price of unit bandwidth decided by the RSU coalition and the optimal amount of bandwidth to purchase decided by VMUs, VTs are migrated from the premigration RSUs to the selected RSU coalition.

\textit{Step 5 (Update reputation values in the blockchain-assisted reputation rating system): } The reputation values of RSUs are updated by the blockchain-assisted reputation rating system. Firstly, based on the reputation values, RSUs (i.e., miners) are proportionally divided into three groups of A, B, and C levels through reputation rules. Then, each group has distinctive responsibilities. Specifically, 1) The A-level RSU group creates a block and broadcasts it to the B-level RSU group. 2) The B-level RSU group validates the block. If the block is legitimate, the B-level RSU group submits the validated block to the C-level RSU group. 3) The C-level RSUs broadcast the block to all RSUs for data synchronization. Finally, the RSUs in the workgroup are rewarded according to their behaviors through reputation bonus-penalty rules, and the reputation values of RSUs participating in the coalition game are updated in the blockchain system. Note that the Practical Byzantine Fault Tolerance (PBFT) consensus algorithm is utilized in the blockchain system for lightweight consensus\cite{li2020scalable}.

\textit{Step 6 (Establish new connections with RSUs of the coalition): }When VTs are migrated to the RSU coalition successfully, VMUs establish new connections with RSUs of the coalition to access metaverse services\cite{Jinbo}, and the RSUs of the coalition will become new premigration RSUs in the next VT migration.

\begin{figure*}[t]
\centering{\includegraphics[width=0.75\textwidth]{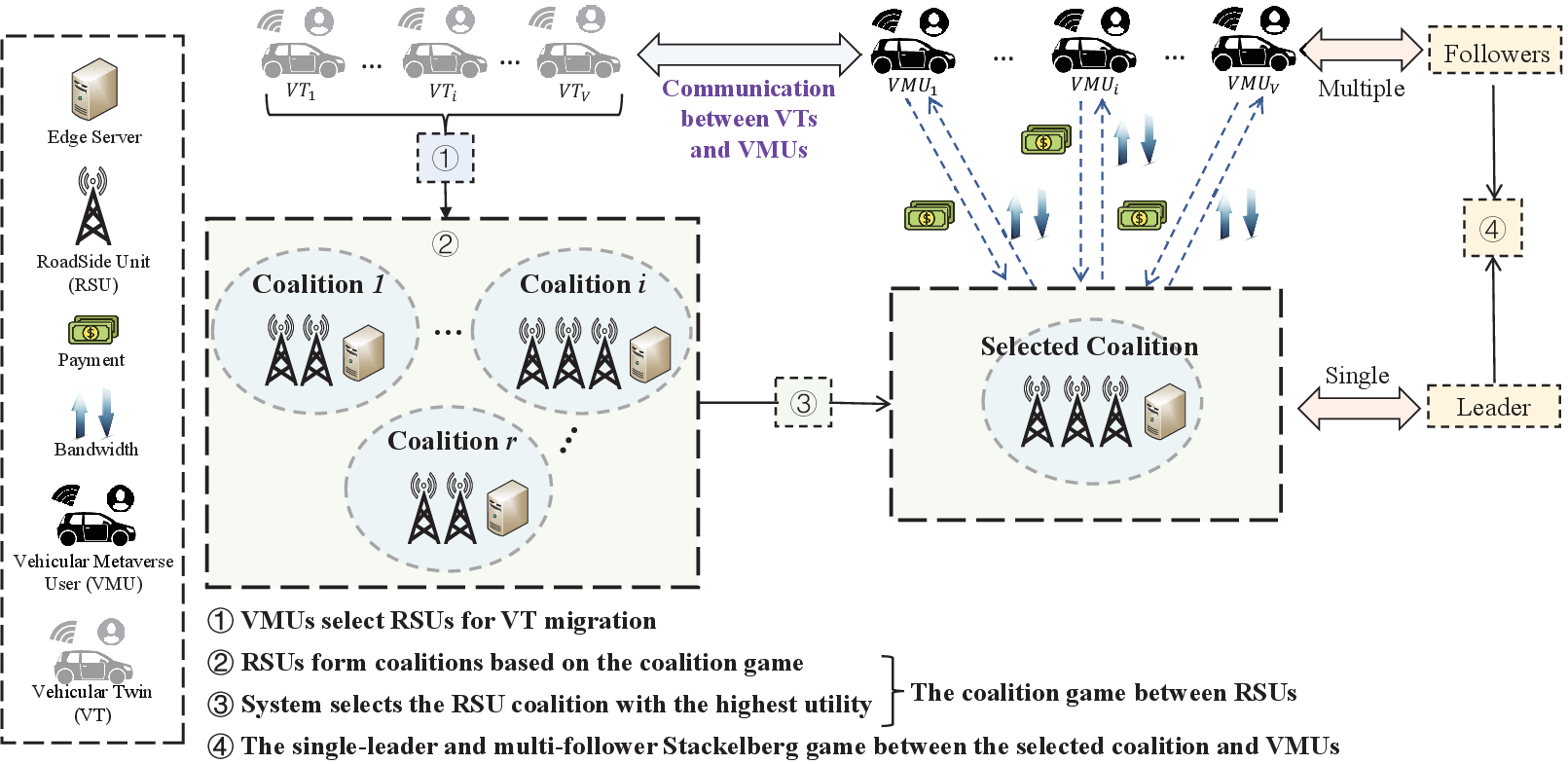}}
\caption{A Stackelberg game between the RSU coalition and VMUs for VT migration.}
\label{Stac}
\end{figure*}

\section{Coalition Game-based RSU Selection for Vehicular Metaverses}\label{Coalition}
\subsection{Subjective Logic Model for Reputation Calculation}
\subsubsection{RSU reputation}
The authors in \cite{josang2006trust} proposed a subjective logic model, the model proposes the concept of fact space and idea space to quantify trust relationships and offers a series of subjective logic operations for trust computation and comprehensive derivation \cite{fu2022bfcri}, and it can quantify trust, doubt and uncertainty, taking into account the credibility of the source of opinion, which is a widely used mathematical tool for reliability modeling \cite{liu2023blockchain}. Therefore, the subjective logic model in computing reputation has been widely utilized \cite{8624307,9829255,8832210}. In this paper, we consider that there is a set of $R$ RSUs, $\Bbb{R}=\{1,\ldots,r,\ldots,R\}$, and a set of $V$ VMUs, $\Bbb{V} = \{1,\ldots,v,\ldots,V\}$. We use a 4-tuple vector $w_{v:r}=\{b_{v:r},d_{v:r},u_{v:r},\alpha_{v:r}\}$ to denote the reputation opinion of VMU $v$ to RSU $r$\cite{9829255}. $b_{v:r}$, $d_{v:r}$, $u_{v:r}$, and $\alpha_{v:r}$ represent belief, disbelief, uncertainty, and the base rate of VMU $v$ toward RSU $r$, respectively, where $b_{v:r},d_{v:r},u_{v:r},\alpha_{v:r}\in{[0,1]}$ and $b_{v:r}+d_{v:r}+u_{v:r}=1$. Belief and disbelief are mapped from the interactions between VMUs and RSUs positively and negatively, respectively \cite{9880566}. The base rate represents the willingness of VMU $v$ to believe RSU $r$, which is an effective uncertainty coefficient of the reputation of RSU $r$ \cite{9880566, 8624307}. We define $R_{v:r}$ as the reputation of VMU $v$ to RSU $r$, which is given by
\begin{equation}
    R_{v:r}=b_{v:r}+\alpha_{v:r}u_{v:r}\label{eq}.
\end{equation}

\subsubsection{Local reputation opinions}
The reputation of RSU $r$ is affected by direct reputation opinions including local reputation opinions\cite{9880566,8832210}. Considering that VMU $v$ may interact with RSU $r$ more than once, the reputation value of VMU $v$ to RSU $r$ is predicted by previous interactions. However, if the interactions between them occurred a long time ago, the reputation value may not have a large effect. Therefore, we define $\tau$ as the effective period for interactions and divide the period $\tau$ into a series of time windows as $\{t_1,\ldots,t_x,\ldots,t_X\}$ \cite{8624307}. The reputation opinion of VMU $v$ to RSU $r$ in the time window $t_x$ is
\begin{equation}
\begin{cases}
b_{v:r}^{t_x}=\frac{\delta_1p_{v:r}^{t_x}}{\delta_1p_{v:r}^{t_x}+\delta_2q_{v:r}^{t_x}+\xi}, \\
d_{v:r}^{t_x}=\frac{\delta_2q_{v:r}^{t_x}}{\delta_1p_{v:r}^{t_x}+\delta_2q_{v:r}^{t_x}+\xi}, \\
u_{v:r}^{t_x}=\frac{\xi}{\delta_1p_{v:r}^{t_x}+\delta_2q_{v:r}^{t_x}+\xi},
\end{cases}
\end{equation}
where $p_{v:r}^{t_x}$ and $q_{v:r}^{t_x}$ are the number of positive and negative interactions between VMU $v$ and RSU $r$ in the time window $t_x$, respectively. $\delta_1$ and $\delta_2$ are weights of the positive interaction and the negative interaction, respectively. $\xi$ is a parameter controlling the rate of uncertainty. To reduce the occurrence of negative interactions \cite{8832210}, we set $\xi=1$, and we ensure that $0<\delta_1\leq\delta_2<1$, satisfying the condition $\delta_1+\delta_2=1$ \cite{7755733}.

The recent interactions with high freshness have a greater impact on the reputation of RSUs than the past interactions. Therefore, it is necessary to consider interaction freshness for the reputation calculation \cite{8832210}. To reflect the influence of time on the reputation calculation, we use $t_{v:r}\in\{T-\tau,T-\tau+1,\ldots,T\}$ to denote the time VMU $v$ interacting with RSU $r$. Then, we use $\mathcal{T}(\cdot)$ to illustrate the degree of reputation attenuation over time, which is given by 
\begin{equation}
    \mathcal{T}(t_{v:r})=\frac{c}{c+\theta(T-t_{v:r})},
\end{equation}
where $\theta\in(0,1)$ is an attenuation coefficient, $c$ is a fixed value that is set to $1$, and $T$ is the current time. By using the time attenuation function to evaluate the reputation, the system can take more into consideration the performance of the RSU in the recent period to identify the reliability of the RSU more accurately. The local reputation opinion of VMU $v$ to RSU $r$ is defined as a vector $\omega_{v:r}^{loc}=\{b_{v:r}^{loc},d_{v:r}^{loc},u_{v:r}^{loc},\alpha_{v:r}^{loc}\}$, which is expressed as \cite{8832210}
\begin{equation}
    \begin{cases}
b_{v:r}^{loc}=\frac{\sum_{t_{v:r}\in\{T-\tau,\ldots,T\}}\mathcal{T}(t_{v:r})b_{v:r}^{t_x}}{\sum_{t_{v:r}\in\{T-\tau,\ldots,T\}}\mathcal{T}(t_{v:r})},\\
d_{v:r}^{loc}=\frac{\sum_{t_{v:r}\in\{T-\tau,\ldots,T\}}\mathcal{T}(t_{v:r})d_{v:r}^{t_x}}{\sum_{t_{v:r}\in\{T-\tau,\ldots,T\}}\mathcal{T}(t_{v:r})}, \\
u_{v:r}^{loc}=\frac{\sum_{t_{v:r}\in\{T-\tau,\ldots,T\}}\mathcal{T}(t_{v:r})u_{v:r}^{t_x}}{\sum_{t_{v:r}\in\{T-\tau,\ldots,T\}}\mathcal{T}(t_{v:r})}, \\
\alpha_{v:r}^{loc}=\alpha_{v:r}.
\end{cases}
\end{equation}
Based on \eqref{eq}, the local reputation of VMU $v$ to RSU $r$ is
\begin{equation}
R_{v:r}^{loc}=b_{v:r}^{loc}+\alpha_{v:r}^{loc}{u_{v:r}^{loc}}.
\end{equation}

\subsubsection{Recommended reputation opinions} In addition to the local reputation opinion of VMU $v$, the reputation value of VMU $v$ to RSU $r$ is also affected by other VMUs acting as recommenders that have interacted with RSU $r$, and their opinions are called recommended reputation opinions\cite{9880566,8832210}. 

We define $\Bbb{M}=\{1,\ldots,m,\ldots,M\} \subset \Bbb{V}$ as a set of recommenders to VMU $v$, meaning that VMU $v$ receives a number of $M$ recommended opinions. The familiarity value between recommender $m$ and RSU $r$ is defined as $F_{m:r}$, which is determined by their interaction frequency. The interaction frequency is the ratio of the number of interactions between recommender $m$ and RSU $r$ to the average number of interactions between recommender $m$ and RSUs \cite{8624307}. Therefore, $F_{m:r}$ is given by
\begin{equation}
    F_{m:r}=\frac{IN_{m:r}}{\overline{IN}_m},
\end{equation}
where $(IN_{m:r}=p_{m:r}+q_{m:r})$ is the interaction numbers between recommender $m$ and RSU $r$ within the interaction period $\tau$, $p_{m:r}$ and $q_{m:r}$ are the number of the positive interactions and the negative interactions between recommender $m$ and RSU $r$, respectively, and $\overline{IN}_m=\frac{\sum_{r\in{\Bbb{R}}}IN_{m:r}}{R}$ \cite{8624307}. Therefore, the reputation of the recommended opinion is $\gamma_{m:r}=\rho_mF_{m:r}$, where $\rho_m\in[0,1]$ is a predefined parameter for reputation calculation \cite{8624307}.

Considering that the recommender's familiarity with RSU $r$ can better use the existing information to reduce the uncertainty value \cite{9829255}, we use $\omega_{m:r}^{rec}=\{b_{m:r}^{rec},d_{m:r}^{rec},u_{m:r}^{rec},\alpha_{m:r}^{rec}\}$ to denote the recommended reputation opinion of VMU $m$ to RSU $r$, where $b_{m:r}^{rec}$, $d_{m:r}^{rec}$, and $u_{m:r}^{rec}$ are given by \cite{8832210}
\begin{equation}
\begin{cases}
b_{m:r}^{rec}=\frac{\sum_{m\in{\Bbb{M}}}\gamma_{m:r}b_{m:r}^{loc}}{\sum_{m\in{\Bbb{M}}}\gamma_{m:r}}, \\
d_{m:r}^{rec}=\frac{\sum_{m\in{\Bbb{M}}}\gamma_{m:r}d_{m:r}^{loc}}{\sum_{m\in{\Bbb{M}}}\gamma_{m:r}}, \\
u_{m:r}^{rec}=\frac{\sum_{m\in{\Bbb{M}}}\gamma_{m:r}u_{m:r}^{loc}}{\sum_{m\in{\Bbb{M}}}\gamma_{m:r}}.
\end{cases}
\end{equation}

\subsubsection{Final reputation opinions}
Based on the above analyses for calculating the local reputation opinion and recommended reputation opinions of RSU $r$, we can further calculate the final reputation opinion of RSU $r$. We use $\omega_{v:r}^{fin}=\{b_{v:r}^{fin},d_{v:r}^{fin},u_{v:r}^{fin},\alpha_{v:r}^{fin}\}$ to denote the final reputation opinion of RSU $r$, which is given by 
\begin{equation}
    \begin{cases}
b_{v:r}^{fin}=\frac{b_{v:r}^{loc}u_{m:r}^{rec}+b_{m:r}^{rec}u_{v:r}^{loc}}{u_{v:r}^{loc}+u_{m:r}^{rec}-u_{v:r}^{loc}u_{m:r}^{rec}}, \\ 
d_{v:r}^{fin}=\frac{d_{v:r}^{loc}u_{m:r}^{rec}+d_{m:r}^{rec}u_{v:r}^{loc}}{u_{v:r}^{loc}+u_{m:r}^{rec}-u_{v:r}^{loc}u_{m:r}^{rec}}, \\
u_{v:r}^{fin}=\frac{u_{v:r}^{loc}u_{m:r}^{rec}}{u_{v:r}^{loc}+u_{m:r}^{rec}-u_{v:r}^{loc}u_{m:r}^{rec}}, \\
\alpha_{v:r}^{fin}=\alpha_{v:r}.
\end{cases}
\end{equation}
Based on \eqref{eq}, the final expectation of the reputation of VMU $v$ to RSU $r$ is expressed as
\begin{equation}
R_{v:r}^{fin}=b_{v:r}^{fin}+\alpha_{v:r}^{fin}{u_{v:r}^{fin}}.
\end{equation}

Without loss of generality, the number of interactions between VMUs and RSUs is set to $0$ on initialization \cite{9880566}. After calculating the final reputation opinions, we can select RSUs with high final reputation values to form coalitions.

\subsection{Coalition Formation Game Formulations}
Due to the limited bandwidth resources of a single RSU, it is not feasible to facilitate simultaneous migration of multiple VTs. As a solution, we propose a coalition game approach for ensuring reliable and large-scale VT migration. In this game, RSUs form coalitions, and the coalition with the highest utility is chosen to allocate bandwidth resources to VMUs. Consequently, this enables the coalition to facilitate the concurrent migration of multiple VTs.

We denote the coalition of RSUs as $\mathcal{G}_o\subseteq{\Bbb{R}}$, and $o$ is an index of the coalition. A group of mutually disjoint coalitions in $\Bbb{R}$ is represented as $\Pi=\{\mathcal{G}_1,\ldots,\mathcal{G}_o,\ldots,\mathcal{G}_O\}$, where $\mathcal{G}_o\neq\mathcal{G}_{o^{'}}$ if $o\neq{o'}$ and $O$ is the number of RSU coalitions \cite{9880566}. Therefore, this coalition game model is made as $\Bbb{G}=\{\Bbb{R},\Pi, \mathcal{U}\}$, where $\mathcal{U}$ represents the utility function of the RSU coalition. The final reputation of RSU $r$ is defined as the average reputation of all VMUs toward RSU $r$, which is expressed as
\begin{equation}
    R^{fin}_r=\frac{\sum_{v=1}^{V}R^{fin}_{v:r}}{V}.
\end{equation}

Because of the traffic volume in some areas (e.g., crossroad areas), RSUs need to be deployed in large numbers, and an edge server may serve multiple RSUs \cite{9124705}. Therefore, we use the RSU node to denote a node composed of an edge server and several RSUs served by it. Then, the RSU nodes construct RSU coalitions. The set of RSU nodes is denoted as $\Bbb{N}=\{1,\ldots,n,\ldots,N\}$. Each RSU has a unique identity number. We denote several RSUs forming an RSU node $n$ as $\Bbb{R}_n=\{\mathcal{R}_n^1,\ldots\mathcal{R}_n^{\omega},\ldots,\mathcal{R}_n^{|\Bbb{R}_n|}\}\subseteq \Bbb{R}$, where $\mathcal{R}_n^{\omega}$ represents the RSU with the identity number $\omega$, which is one of the components of the RSU node. Therefore, the RSU set for the coalition $\mathcal{G}_o$ is denoted as $\cup_{n\in{\mathcal{G}_o}}\Bbb{R}_n$, the number of RSUs in $\cup_{n\in{\mathcal{G}_o}}\Bbb{R}_n$ is denoted as $|\cup_{n\in{\mathcal{G}_o}}\Bbb{R}_n|$, the RSU node set for the coalition $\mathcal{G}_o$ is denoted as $\Bbb{N}_o$, and the number of RSU nodes in $\Bbb{N}_o$ is denoted as $|\Bbb{N}_o|$.

Based on the calculation of RSUs' reputation opinions, the contribution value of the coalition $\mathcal{G}_o$ is expressed as \cite{9880566}
\begin{equation}
\mathcal{Q}(\mathcal{G}_o)=\zeta_{1}\frac{|\cup_{n\in{\mathcal{G}_o}}\Bbb{R}_n|}{R}+\zeta_{2}\frac{\sum_{n\in{\mathcal{G}_o}}\sum_{r=1}^{|\Bbb{R}_n|}R_{r}^{fin}}{|\cup_{n\in{\mathcal{G}_o}}\Bbb{R}_n|},
\end{equation}
where the first part $\frac{|\cup_{n\in{\mathcal{G}_o}}\Bbb{R}_n|}{R}$ is the percentage of RSUs that the coalition $\mathcal{G}_o$ has, and the second part $\frac{\sum_{n\in{\mathcal{G}_o}}\sum_{r=1}^{|\Bbb{R}_n|}R_{r}^{fin}}{|\cup_{n\in{\mathcal{G}_o}}\Bbb{R}_n|}$ is the average reputation value of RSUs in the coalition $\mathcal{G}_o$. $\zeta_1$ and $\zeta_2$ represent the weights of two parts of the contribution value \cite{9880566}.

Since the large latency of VT migration leads to a poor immersive experience for VMUs, in addition to considering the contribution value, we also consider the latency of VT migration for the utility of the coalition. Each RSU joining the coalition can decide the amount of bandwidth provided for VT migration, and the provided bandwidth of the coalition is the sum of the bandwidth provided by RSU nodes in the coalition. We define the provided bandwidth of the coalition as $B$ and the communication rate of the coalition as $R_{t}$. Based on the Shannon theorem \cite{9973495}, $R_t$ is given by 
\begin{equation}
    R_t=B\log_2\bigg(1+\frac{\rho{h^0}d^{-\varepsilon}}{N_0}\bigg),
    \label{Rt}
\end{equation}
where $\rho$, $h^0$, $d$, $\varepsilon$, and $N_0$ represent the transmitter power of the premigration RSU, the unit channel power gain, the average distance between RSUs, the path-loss exponent, and the average noise power, respectively \cite{9973495}. Note that $R_t$ is variable because the total amount of bandwidth $B$ is a variable decided by the RSUs participating in the coalition formation. Therefore, the service latency of the coalition $\mathcal{G}_o$ is  \cite{5289170}
\begin{equation}
    \mathcal{I}(\mathcal{G}_o)=\frac{D\lambda}{R_t},\label{I_t}
\end{equation}
where $D$ is the VT data size of the VMU and $\lambda$ is the data compression ratio.

Since forming an RSU coalition needs negotiation and information exchange between RSU nodes, which may incur communication costs and reduce the utilities of RSU coalitions \cite{9880566,5285181}, the communication cost between RSU nodes is needed to be considered, which is denoted as $\mathcal{C}({\mathcal{G}_o})$. Communication cost calculation should meet two conditions \cite{9880566}. The first condition is that $\mathcal{C}({\mathcal{G}_o})$ increases with the increase of the number of RSU nodes $|\Bbb{N}_o|$. The second condition is that the slope of $\mathcal{C}({\mathcal{G}_o})$ becomes steeper with the increase of $|\Bbb{N}_o|$. Therefore, the communication cost of the coalition $\mathcal{G}_o$ is expressed as \cite{9880566}
\begin{equation}
\mathcal{C}(\mathcal{G}_o)=
\begin{cases}
    -\mathrm{ln}\Big(1-\frac{|\Bbb{N}_o-\epsilon|}{N}\Big),\:|\Bbb{N}_o|\ge2, \\
    0  \hspace{2.2cm},\:\mathrm{otherwise},
\end{cases}
\end{equation}
where $\epsilon$ is set to $0.1$, which is used to avoid an infinite value of $\mathcal{C}(\mathcal{G}_o)$ when $|\Bbb{N}_o|=N$. Based on the contribution value $\mathcal{Q}(\mathcal{G}_o)$, the service latency $\mathcal{I}(\mathcal{G}_o)$, and the communication cost $\mathcal{C}(\mathcal{G}_o)$, the utility function of the coalition $\mathcal{G}_o$ is expressed as 
\begin{equation}
    \mathcal{U}(\mathcal{G}_o)=\mathcal{Q}(\mathcal{G}_o)+\gamma{\mathrm{ln}\bigg(1+\frac{1}{\mathcal{I}(\mathcal{G}_o)}\bigg)}-\sigma{\mathcal{C}(\mathcal{G}_o)}
    \label{eq2},
\end{equation}
where $\gamma$ and $\sigma$ are coefficients that represent the service latency and communication cost, respectively.

\begin{definition}
$\textbf{(Non-Transferable Utility (NTU)):}$ Let $\psi(\cdot)$ is a mapping function such that for every coalition $\mathcal{G}_o\subseteq{\Bbb{R}}$, $\psi(\mathcal{G}_o)$ is a closed convex subset of $\Bbb{R}^{\mathcal{G}_o}$ that contains the utility vectors that RSUs in $\mathcal{G}_o$ can achieve \cite{5607318}.
\end{definition}

For the coalition $\mathcal{G}_o$, whether the RSUs (i.e., the miners) can be rewarded by the blockchain-assisted reputation rating system depends on the coalition utility $\mathcal{U}(\mathcal{G}_o)$\cite{9880566}. If a coalition game is said to be NTU, the utility of a coalition cannot arbitrarily be divided between coalition members\cite{5607318}. Therefore, the RSU coalition selection can be modeled as a coalition formation game with NTU. Each RSU can choose the suitable coalition based on the received utility, and the utility of each RSU in the coalition $\mathcal{G}_o$ is equal to $\mathcal{U}(\mathcal{G}_o)$ instead of a fraction of $\mathcal{U}(\mathcal{G}_o)$\cite{9880566}.

\subsection{Coalition Formation with Merge-and-Split Rules}
\begin{definition}
$\textbf{(Preference operator):}$ A preference operator $\rhd$ is defined for comparing $\Pi_1=\{\mathcal{G}_1^1,\ldots,\mathcal{G}_O^1\}$ and $\Pi_2=\{\mathcal{G}_1^1,\ldots,\mathcal{G}_{O'}^2\}$ that are partitions of the same subset $\Bbb{A}\subseteq\Bbb{R}$ (i.e., same RSUs in $\Pi_1$ and $\Pi_2$). Therefore, $\Pi_1\rhd\Pi_2$ represents that $\Pi_1$ is better than $\Pi_2$ for subset $\Bbb{A}$ \cite{9880566,5285181}.
\end{definition}

The authors in \cite{ding2022coalition} proposed an approach for coalition formation based on merge and split rules. Many orders can be used to compare relationships between partitions, e.g., coalition value orders and individual value orders \cite{5285181}. Individual value orders compare relationships between partitions by using the individual payoff, e.g., Pareto order \cite{5285181}. In this paper, we use the Pareto order to perform the comparison.  

\begin{definition}
    $\textbf{(Pareto order):}$ For two partitions $\Pi_1=\{\mathcal{G}_1^1,\ldots,\mathcal{G}_O^1\}$ and $\Pi_2=\{\mathcal{G}_1^1,\ldots,\mathcal{G}_{O'}^2\}$, the utility of RSU $r$ in $\Pi_1$ and $\Pi_2$ are denoted as $\mathcal{U}_r(\Pi_1)$ and $\mathcal{U}_r(\Pi_2)$, respectively. Then, $\pi_1$ is better than $\pi_2$ with the Pareto order defined as 
\begin{equation}
    \Pi_1\rhd\Pi_2\Longleftrightarrow\{{\mathcal{U}_r(\Pi_1)\ge{\mathcal{U}_r(\Pi_2)}},\:\forall{r}\in\{\Pi_1,\Pi_2\},
\end{equation}
with at least one strict inequality$(>)$ for RSU $r$ \cite{9880566,5607318}.
\end{definition}

For the same RSUs, the partition $\Pi_1$ is preferred over the partition $\Pi_2$ by the Pareto order if at least one RSU can improve its utility when it joins $\Pi_1$ from $\Pi_2$ without reducing the utility of other RSUs. We adopt a coalition formation algorithm based on the Parote order utilized for comparison and the merge-and-split rules \cite{9880566}. Coalition formation requires multiple rounds of merging and splitting, involving all coalitions in each round\cite{9880566}. It is necessary to ensure that the utilities of all coalitions remain stable or increase during the formation process. The merge-and-split rules are defined as follows \cite{9880566,9973643}:
\begin{itemize}
\item \textbf{Merge Rule:} For any set of coalitions $\{\mathcal{G}_1,\ldots,\mathcal{G}_O\}$, merge $\{\mathcal{G}_1,\ldots,\mathcal{G}_O\}$ into $\{\cup_{o=1}^O\mathcal{G}_o\}$, i.e., $\{\cup_{o=1}^O\mathcal{G}_o\}\rhd\{\mathcal{G}_1,\ldots,\mathcal{G}_O\}$, which is denoted as $\{\mathcal{G}_1,\ldots,\mathcal{G}_O\}\rightarrow\{\cup_{o=1}^O\mathcal{G}_o\}$.
\item \textbf{Split Rule:} For any set of coalitions $\{\mathcal{G}_1,\ldots,\mathcal{G}_O\}$, split $\{\cup_{o=1}^O\mathcal{G}_o\}$ into $\{\mathcal{G}_1,\ldots,\mathcal{G}_O\}$, i.e., $\{\mathcal{G}_1,\ldots,\mathcal{G}_O\}\rhd\{\cup_{o=1}^O\mathcal{G}_o\}$, which is denoted as $\{\cup_{o=1}^O\mathcal{G}_o\}\rightarrow\{\mathcal{G}_1,\ldots,\mathcal{G}_O\}$.
\end{itemize}

The merge-and-split rules based on the Pareto order indicate that a coalition will merge only if at least one RSU can improve its utility by merging without reducing the utilities of other RSUs. Similarly, a coalition splits only if at least one RSU can improve its utility by splitting without harming other RSUs, namely reducing the utilities of other RSUs \cite{9880566}. With the merge-and-split rules based on the Pareto order, we propose a coalition algorithm based on \cite{9880566, 5285181} to form RSU coalitions, which consists of three main phases: initialization, adaptive coalition formation, and selection. In the initialization phase, all RSUs are disjoint, and they form initial coalitions. In the adaptive coalition formation phase, the merge-and-split rules based on the Pareto order are utilized to form coalitions by maximizing the utilities of all coalitions. In the selection phase, the RSU coalition with the highest utility is selected. Here are the details of the coalition algorithm.

\begin{algorithm}[t]
\caption{Coalition Formation Algorithm}
\begin{algorithmic}
\Require The RSU set $\Bbb{R}=\{1,\ldots,r,\ldots,R\}$, the final reputation of each RSU $R^{fin}_r,\:r\in{\Bbb{R}}$, and the RSU node set $\Bbb{N}=\{1,\ldots,n,\ldots,N\}$.
\Ensure The RSU coalition with the highest utility.
\State \textit{\textbf{Phase 1 - Initialization}}
\State {Initialize the RSU partition, i.e., each RSU forms a coalition, indicating that all RSUs are disjoint.}   
\State \textit{\textbf{{Phase 2 - Adaptive Coalition Formation}}}
\State {Compute the utility of each coalition based on the utility function \eqref{eq2}.}
\Repeat
\State {Merge mechanism: The coalition $\mathcal{G}_o$ merges into $\mathcal{G}_{o^\prime}$ according to the \textbf{Merge Rule} \cite{9880566,5285181,9973643}.}
\State {Split mechanism: The coalition $\mathcal{G}_o$ splits into $\mathcal{G}_{o^\prime}$ according to the \textbf{Split Rule} \cite{9880566,5285181,9973643}.}
\Until{{Merge-and-split iteration terminates.}}
\State \textit{\textbf{{Phase 3 - Selection}}}
\State {The RSU coalition with the highest utility is selected.}
\end{algorithmic}
\end{algorithm}

The computational complexity of \textbf{Algorithm 1} mainly depends on the merge-and-split process \cite{9880566}. Note that the worst computational complexity is $\mathcal{O}(R^3)$\cite{9880566}. A large coalition is formed after the first merge operation. The split operation is only performed on each RSU coalition, reducing the complexity of the split operation. The result of \textbf{Algorithm 1} is an RSU partition consisting of disjoint independent confederations. The stability of the final coalition partition can be analyzed by the defection function $\Bbb{D}_{hp}$.

\begin{definition}
\textbf{(Defection function $\Bbb{D}_{hp}$): }A partition $\Pi=\{\mathcal{G}_1,\ldots,\mathcal{G}_o,\ldots,\mathcal{G}_O\}$ is $\Bbb{D}_{hp}$-stable if no RSU wants to leave $\Pi$, or when the RSUs which want to leave only can form the partitions allowed by $\Bbb{D}_{hp}$ \cite{9827991,5285181}.    
\end{definition}

If $\Pi=\{\mathcal{G}_1,\ldots,\mathcal{G}_o,\ldots,\mathcal{G}_O\}$ is $\Bbb{D}_{hp}$-stable, two conditions need to be satisfied \cite{5285181}: 
\begin{enumerate}[i)]
\item  For $o\in{\{1,\ldots,O\}}$ and each partition $\{\Bbb{R}_1,\ldots,\Bbb{R}_p\}$ of coalition $\mathcal{G}_o:\{\Bbb{R}_1,\ldots,\Bbb{R}_p\}\ntriangleright\mathcal{G}_o$, where $\ntriangleright$ is the opposite rule of $\vartriangleright$.
\item  For $S\in{\{1,\ldots,O\}}:\bigcup_{o\in{S}}\mathcal{G}_o\ntriangleright\{\mathcal{G}_o|o\in{S}\}$ \cite{5285181}.
\end{enumerate} 
\begin{theorem}
The final partition resulting from our coalition formation algorithm based on merge-and-split rules is $\Bbb{D}_{hp}$-stable.
\end{theorem}

\begin{proof}
Please refer to \cite{9880566}.
\end{proof}

\section{Stackelberg Model for Vehicular Metaverses}\label{Stackelberg}
After the coalition game, the RSU coalition with the highest utility will be selected. The upper limit of provided bandwidth of the RSU coalition depends on the amount of bandwidth each RSU node in the coalition is willing to contribute, which is denoted as $B_{max}$. We consider that each VMU has a corresponding VT for managing vehicular applications and VTs would be migrated from the premigration RSUs to the RSU coalition. The RSU coalition can determine the selling price of unit bandwidth and VMUs determine the amount of the purchased bandwidth based on the price unit of bandwidth. We define $B_v$ as the amount of bandwidth that the RSU coalition provides for VMU $v$. The RSU coalition can earn $P$ per unit of bandwidth from each VMU. Simultaneously, the RSU coalition needs to pay the transmission cost of $C$ of unit bandwidth.

The RSU coalition can determine how much VMUs should pay for bandwidth, and based on the bandwidth price, VMUs can determine how much bandwidth they would purchase. Therefore, we formulate a single-leader multi-follower Stackelberg model between the RSU coalition and VMUs, which is denoted as $\mathcal{G}$. In the Stackelberg model, the RSU coalition acting as a leader first declares its strategy, i.e., the price of unit bandwidth. Based on the leader's strategy, VMUs acting as followers would decide their strategies, i.e., the amount of bandwidth requested. The Stackelberg game model is described in detail as follows:

\subsection{VMUs' Bandwidth Service Demands in Stage II}
In this part, we formulate the utility function of the VMU. For VMU $v$, we define $A_v$ as the service latency of the VT migration, where $D_v$ is defined as the data size of the VT. Similar to (\ref{I_t}), the service latency of the VT migration of VMU $v$ is given by
\begin{equation}
A_v=\frac{D_v\lambda}{R_t}=\frac{D_v\lambda}{B_v\log_2\big(1+\frac{\rho{h^0}d^{-\varepsilon}}{N_0}\big)}.
\end{equation}
The higher the bandwidth price or the service latency, the lower profits that VMUs obtain. However, the lower the bandwidth price set by the RSU coalition, the larger the response time of bandwidth it provides, so VMUs need to decide their strategies based on the RSU coalition's strategy. Therefore, the utility function of VMU $v$ is expressed as \cite{9332231}
\begin{equation}
    U_v=\alpha_v{\mathrm{ln}\Big(1+\frac{1}{A_v}\Big)}-PB_v,
    \label{U_v}
\end{equation}
where $\alpha_v\in(0,1)$ is a parameter centered on VMU 
$v$, indicating the sensitivity of VMU $v$ to the service latency of VT migration.

\subsection{RSU Coalition's Selling Price in Stage I}
The RSU coalition acting as a bandwidth resource provider not only ensures that its resource allocation can meet the needs of VMUs but also ensures that its utility can be maximized. To incentivize VTs to be migrated to the RSU coalition and gain as much as possible profits, the RSU coalition formulates an appropriate pricing strategy, indicating that the RSU coalition needs to constantly adjust its pricing strategy according to the bandwidth demands of VMUs to maximize its utility. The RSU coalition can obtain profits by providing bandwidth resources to VMUs but needs to pay the transmission costs of bandwidth resources. Therefore, the problem of maximizing the utility of the RSU coalition is formulated as
\begin{equation}
    \begin{split}
    \textbf{Problem:}\:&\max\limits_{P}\:U_r=\sum_{v=1}^V(P-C)B_v,  \\
    &\:\:s.t.\:\: {0 < \textstyle \sum_{v=1}^{V}}B_v \leq B_{max},\\
    &\quad\:\:\:\:\:\: 0 < B_v,\:\forall v \in \small\{1,\ldots,V\small\},\\
    &\quad\:\:\:\:\:\: 0 < C \leq P \leq P_{max} .
    \end{split}
    \label{U_r}
\end{equation}
where $U_r$ is the utility of the RSU coalition, $B_{max}$ is the total amount of bandwidth that the RSU coalition provides, which has been determined by the selected RSU coalition after the coalition game, and $P_{max}$ is the maximum selling price of unit bandwidth determined by the RSU coalition. Note that no VMU would buy bandwidth resources from the RSU coalition if the selling price of unit bandwidth exceeds $P_{max}$.

\subsection{Stackelberg Equilibrium Analysis}
The Stackelberg equilibrium ensures that the utility of the RSU coalition is maximized, considering that VMUs formulate policies of requesting the amount of bandwidth according to the best response. In this part, we seek the Stackelberg equilibrium, at which the RSU coalition acts as a leader and VMUs act as followers. Both the leader and followers can maximize their utilities by constantly changing their strategies until they reach the optimal strategies in equilibrium. The Stackelberg equilibrium is defined as follows:

\begin{definition}
    \textbf{(Stackelberg Equilibrium):}  Let $\boldsymbol{B}^*=\{B_v^*\},\:v\in{\Bbb{V}}$ and $P^*$ are denoted as the optimal bandwidth demands of VMUs and optimal pricing bandwidth of the RSU coalition, respectively. The strategy $(\boldsymbol{B}^*,P^*)$ is the Stackelberg equilibrium if and only if the following set of inequalities are strictly satisfied \cite{9973495}
\begin{equation}
    \left\{\begin{array}{l}U_r(P^*,\boldsymbol{B}^*)\ge{U_r(P,\boldsymbol{B}^*)},\vspace{0.05in}\\ 
     U_v(B_v^*,\boldsymbol{B}_{-v}^*,P^*)\ge{U_v(B_v,\boldsymbol{B}_{-v}^*,P^*)},\:\forall{v\in\Bbb{V}}.\end{array}\right.
\end{equation}
\end{definition}

In the following, we utilize the backward induction method to analyze the Stackelberg equilibrium \cite{9880566}. 

\subsubsection{VMUs' optimal strategies as equilibrium in Stage II}
In the Stackelberg game, based on the selling price of unit bandwidth $P$, VMUs acting as followers would determine the optimal bandwidth demand strategies to maximize their profits in \textit{Stage II} \cite{9880566}.

\begin{theorem}
   The sub-game perfect equilibrium in the VMUs' subgame is unique \cite{9880566,9860983}. 
\end{theorem}
\begin{proof}
    The first-order derivative and the second-order derivative of $U_v$ with respect to $B_v$ are shown as
\begin{equation}
\begin{split}
    \frac{\partial{U_v}}{\partial{B_v}}&=\frac{\alpha_v\log_2\Big(1+\frac{\rho{h^0}d^{-\varepsilon}}{N_0}\Big)}{D_v\lambda+B_v\log_2\big(1+\frac{\rho{h^0}d^{-\varepsilon}}{N_0}\big)}-P,\\
    \frac{\partial^2{U_v}}{\partial{B_v}^2}&=-\frac{\alpha_v\log_2\Big(1+\frac{\rho{h^0}d^{-\varepsilon}}{N_0}\Big)^2}{\big(D_v\lambda+B_v\log_2(1+\frac{\rho{h^0}d^{-\varepsilon}}{N_0})\big)^2}<0.
\end{split}
\end{equation}
Since the first-order derivative of $U_v$ has a unique zero point and the second-order derivative of $U_v$ is negative, the utility function $U_v$ is strictly concave concerning the VMU's bandwidth demand strategy $B_v$. Therefore, the sub-game perfect equilibrium in the VMUs’ subgame is unique.
\end{proof}
Then, we set the first-order derivative of $U_v$ to $0$, and get the best response function $B_v^*$ for VMU $v$, which is given by
\begin{equation}
    B_v^*=\frac{\alpha_v}{P}-\frac{D_v\lambda}{\log_2\big(1+\frac{\rho{h^0}d^{-\varepsilon}}{N_0}\big)}.
    \label{Bv}
\end{equation}

\subsubsection{RSU coalition's optimal strategy as equilibrium in Stage I}
To analyze the existence and uniqueness of the equilibrium of the Stackelberg game, we study the concavity of the utility function of the RSU coalition. By predicting the strategies of VMUs, the RSU coalition plays as the leader to maximize its utility in \textit{Stage I}.

\begin{theorem}
    The uniqueness of the proposed Stackelberg game equilibrium can be guaranteed.
\end{theorem}
\begin{proof}
    Based on the optimal bandwidth demand strategies of VMUs, the utility function of the RSU coalition is given by
\begin{equation}
    U_r=\sum_{v=1}^V(P-C)\Bigg(\frac{\alpha_v}{P}-\frac{D_v\lambda}{\log_2\big(1+\frac{\rho{h^0}d^{-\varepsilon}}{N_0}\big)}\Bigg).
\end{equation}
Taking the first-order derivative and second-order derivative of $U_r$ with respect to $P$, we have
\begin{equation}
\begin{split}
    \frac{\partial{U_r}}{\partial{P}}&=\sum_{v=1}^V\Bigg(\frac{\alpha_vC}{P^2}-\frac{D_v\lambda}{\log_2\big(1+\frac{\rho{h^0}d^{-\varepsilon}}{N_0}\big)}\Bigg),\\
    {\frac{\partial^2{U_r}}{\partial{P}^2}}&=-\sum_{v=1}^V\frac{2\alpha_vC}{P^3}<0.
\end{split}
\end{equation}
Similarly, since the first-order derivative of $U_r$ has a unique zero point and the second-order derivative of $U_r$ is negative, the utility function of the RSU coalition is concave, indicating that the RSU coalition has a unique optimal solution. Therefore, the RSU coalition has a unique optimal strategy and the uniqueness of the Stackelberg game's equilibrium is proved. 
\end{proof}

\begin{algorithm}[t]
\caption{Iterative Algorithm for Seeking Stackelberg Equilibrium}
\begin{algorithmic}
    \Require $C, P_{max}$, $B_{max}, \alpha_v,D_v,\forall{v}\in{\Bbb{V}}$.
    \Ensure The optimal pricing strategy $P^*$ and the optimal bandwidth demand strategies $\boldsymbol{B}^*$.
    \State Initialize $U_r^*=0$, $P^*=0$;
    \For{$P=C$ \textbf{to} $P_{max}$}
    \State Calculate $B_v$ based on \eqref{Bv};
    \If {$\sum_{v=1}^V{B_v} \leq B_{max}$}
    \State CALCULATE$(U_r, U_v)$;
    \If{$U_r > U_r^*$ \textbf{and} $U_v > 0$}
    \State Replace $U_r^*$ with $U_r$;
    \State Replace $P^*$ with $P$;
    \EndIf
    \EndIf
    \EndFor
    \State Calculate $\boldsymbol{B}^*$ based on \eqref{Bv};
    \Function {CALCULATE}{$U_r, U_v$}
    \State Calculate the utility of the RSU coalition $U_r$ based on \eqref{U_r};
    \State Calculate the utility of the VMU $U_v$ based on \eqref{U_v};
    \EndFunction
\end{algorithmic}
\end{algorithm}

\begin{table}[t]\label{parameter}
  \begin{center}
    \caption{Key Parameters in the Simulation.}
    \begin{tabular}{l|r} 
    \toprule 
      \textbf{Parameters} & \textbf{Values}\\
      \hline
      Positive interaction frequency/(min) & $[0,100]$  \\
      Negative interaction frequency/(min) & $[0,200]$  \\
      The weight of positive interactions
      $\delta_1$ & $0.5$\\
      The attenuation coefficient $\theta$ & $0.5$ \\
      Reputation threshold $T_{th}^{fin}$ & $0.5$ \\
      Data compression ratio $\lambda$ & $0.5$ \\
      Path-loss exponent $\varepsilon$ & $2$ \\
      Transmitter power of the premigration RSU $\rho$ & $40\rm{dBm}$ \\
      Unit channel power gain $h^0$ & $-20\rm{dB}$ \\
      The average distance between RSUs $d$ & $500\rm{m}$ \\
      Average noise power $N_0$ & $-150\rm{dB}$ \\
      The VT data size of the VMU $D_v$  & $500\rm{MB}$ \\
      The maximum price of unit bandwidth $P_{max}$ & $100$\\
      \bottomrule
    \end{tabular}\label{parameter}
  \end{center}
\end{table}

Motivated by the above analyses, we propose an iterative algorithm to find the Stackelberg equilibrium, as shown in \textbf{Algorithm 2}. The computational complexity of \textbf{Algorithm 2} is $\mathcal{O}\Big(V\big(\frac{P_{max}-C}{\varphi}\big)\Big)$. At first, we initialize the basic parameters. Especially, the optimal strategy $P^*$ and the highest utility of the RSU coalition $U_r^*$ are both initialized as $0$. Then, the selling price of unit bandwidth $P$ is increased by $\varphi$ iteratively and the amount of bandwidth requested by VMUs is calculated in each iteration. If the total amount of bandwidth requested by all VMUs $\sum_{v=1}^VB_v$ does not exceed the maximum amount of bandwidth $B_{max}$, the utilities of VMUs and the RSU coalition can be calculated based on \eqref{U_v} and \eqref{U_r}. When a new optimal value of $U_r$ is found and the utilities of all VMUs are greater than $0$, the value of $U_r$ and the corresponding pricing strategy $P$ of the RSU coalition are recorded. Finally, each VMU can determine its optimal strategy $B_v^*,\:v\in{\Bbb{V}}$ after knowing the final pricing strategy of the RSU coalition.

    \begin{figure}[t]
        \centering{\includegraphics[width=0.45\textwidth]{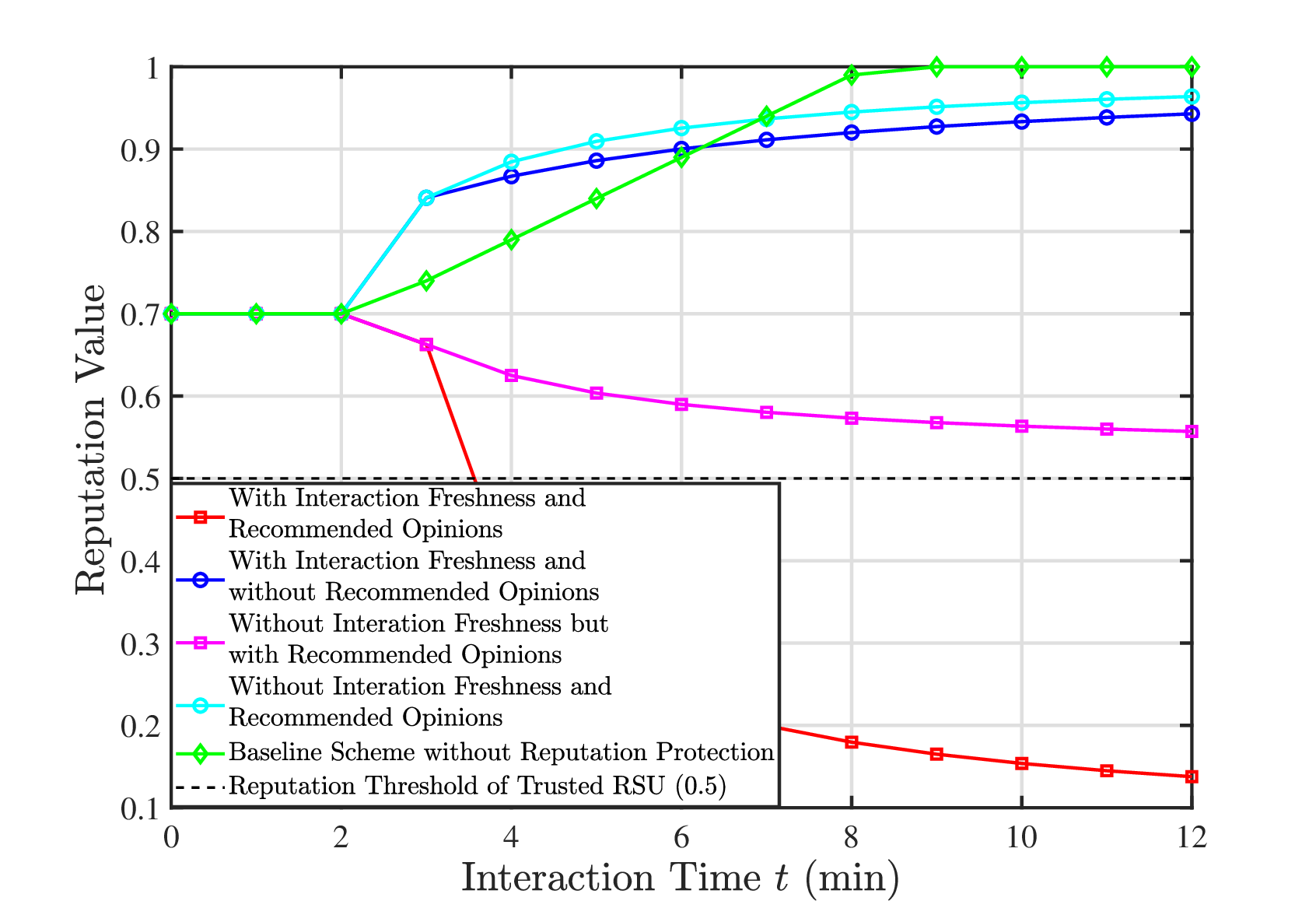}}
        \caption{Reputation value variation of an unreliable RSU.}
        \label{rep_time}
    \end{figure}
    
    \begin{figure}[t]
        \centering
        \includegraphics[width=0.4\textwidth]{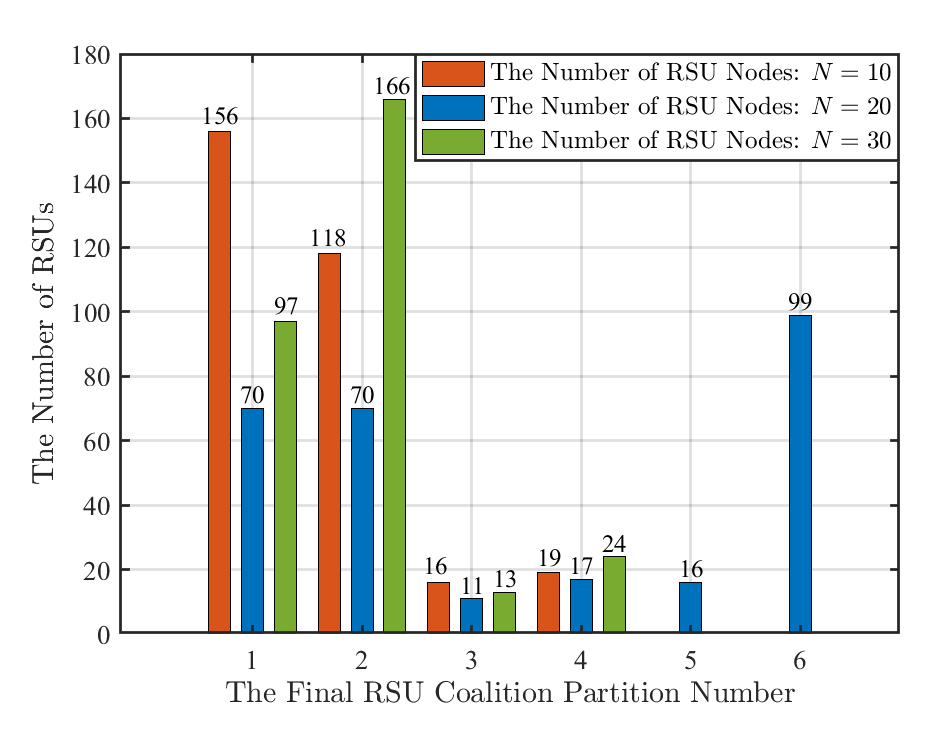}
        \caption{The distribution of the final coalitions under different numbers of RSU nodes.}
        \label{rsu}
    \end{figure}

\section{Numerical Results}\label{Results}
In this section, we present numerical results for blockchain-assisted RSU selection based on the coalition game and the Stackelberg game between the selected RSU coalition and VMUs. Similar to \cite{9880566, Jinbo}, the major parameters are listed in Table \ref{parameter}.

\subsection{Performance of the Proposed Reputation Scheme}
In our proposed reputation scheme, VMUs calculate the reputation of RSUs based on local opinions and recommended opinions. We consider that an unreliable RSU performs well at first to obtain an initialized reputation value of $0.7$, which maintains a good performance within a certain period. Then, the RSU continues to perform well on some VMUs, but poorly on other VMUs with a probability of $90\%$ \cite{9880566}. 
Notably, positive interactions from VMUs enhance the reputation value of the RSU, whereas negative interactions from VMUs diminish the RSU's reputation value.

Figure ~\ref{rep_time} shows the reputation
value changes of an unreliable RSU over time. When the interaction time $t=3$, the RSU begins to misbehave with certain VMUs and gets progressively worse. Note that the effective interaction period $\tau$ is counted from $t=3$. From Fig.~\ref{rep_time}, we can observe that as the unreliable RSU misbehaves over time, the reputation value with interaction freshness and recommended opinions declines quickly below the trusted reputation threshold, which indicates that our proposed reputation scheme can effectively identify unreliable RSUs. Besides, the reputation value decays more rapidly when considering interaction freshness compared to when not considering it.
Therefore, our proposed reputation scheme with interaction freshness can identify unreliable RSUs more efficiently than that without interaction freshness. For the baseline scheme without reputation protection and the reputation scheme without recommended opinions, the reputation value of the unreliable RSU both increases over time \cite{8832210}. The reason is that an unreliable RSU only performs well on several VMUs and performs maliciously on other VMUs. If VMUs interacting with the RSU positively only consider the local opinions without considering the recommended opinions of other VMUs with negative interactions, the RSU will naturally obtain a high reputation value.

\subsection{Numerical Analysis for the Coalition Game}

Figure \ref{rsu} shows the distribution of the final RSU coalitions under different numbers of RSU nodes. We consider that the total number of RSUs is $R=200$. When the final coalitions are certain, the RSU coalition with the highest utility will be selected. From Fig.~\ref{rsu}, we can see that the number of RSUs in each coalition is different when the number of RSU nodes is different. When the number of RSU nodes is $N=20$, there are $6$ final RSU coalitions formed. Similarly, when the number of RSU nodes is $N=10$, there are $4$ final RSU coalitions established. To be specific, the first RSU coalition consists of $156$ RSUs, while the fourth RSU coalition is comprised of $19$ RSUs. Note that regardless of the number of RSU nodes, the total number of RSUs for all coalitions may exceed $200$. The reason is that the composition of RSU nodes has a certain degree of randomness and each RSU can join multiple coalitions, i.e., an RSU can be in more than one coalition.

Figure \ref{fig4} shows the spent time constructing the final RSU coalitions corresponding to the number of RSU nodes under different total numbers of RSUs. From Fig.~\ref{fig4}, we can see that as the number of RSU nodes grows, the time to form the final coalitions through the coalition game increases. In addition, it can be seen that the more the number of RSUs, the larger the time required for constructing the final coalitions in the case of the same number of RSU nodes. When the number of RSU nodes exceeds $20$, the time needed to form the final coalitions with $200$ RSUs is significantly greater compared to the case of $100$ RSUs. To avoid a long time of coalition formation that may affect the quality of vehicular metaverse services to VMUs, machine learning models can be utilized to predict the driving route of vehicles so that RSUs can form coalitions in advance to provide bandwidth resources for VT migration timely. 

Figure \ref{mis_rep} shows the average reputation value corresponding to different misbehavior ratios under different numbers of RSU nodes. The misbehavior ratio is the percentage of RSUs that have negative interactions with VMUs of the total number of RSUs \cite{9880566}. From Fig. \ref{mis_rep}, we can see that with the increase of the misbehavior ratio, the average reputation value of RSUs does not change greatly regardless of the number of RSU nodes. The reason is that with the help of the proposed reputation scheme, the RSUs with low reputation values are excluded from the coalition game. Therefore, RSU nodes can select the RSUs with high reputation values to form coalitions in the coalition game, i.e., our proposed scheme can construct trustworthy RSU coalitions. Then, the RSU coalition with the highest utility will be selected and the maximum amount of bandwidth of the selected RSU coalition will be obtained.
\begin{figure}[t]
\centering
\includegraphics[width=0.4\textwidth]{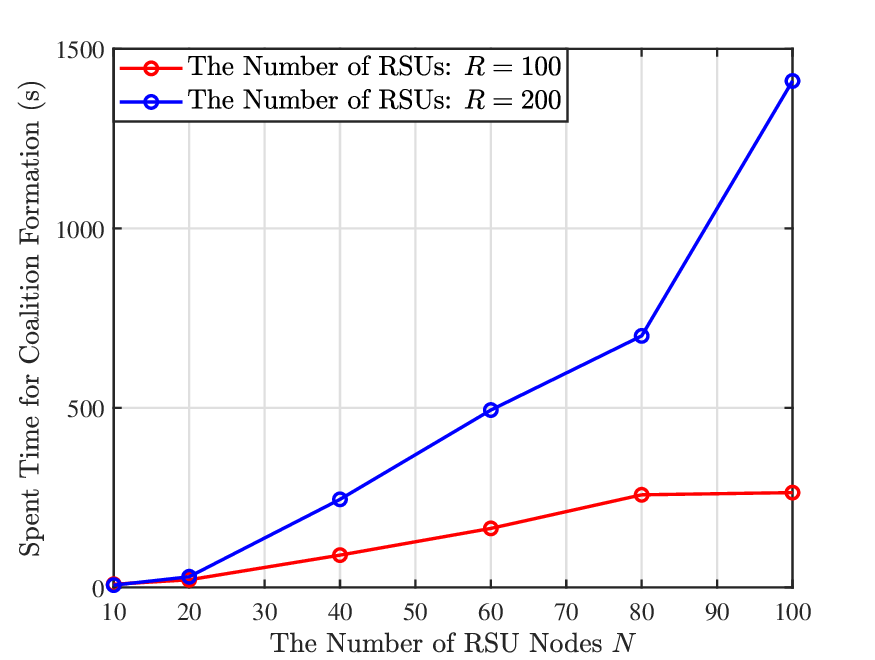}
\caption{Spent time for coalition formation under different numbers of RSUs.}
\label{fig4}
\end{figure}
\begin{figure}[t]
\centering
\includegraphics[width=0.4\textwidth]{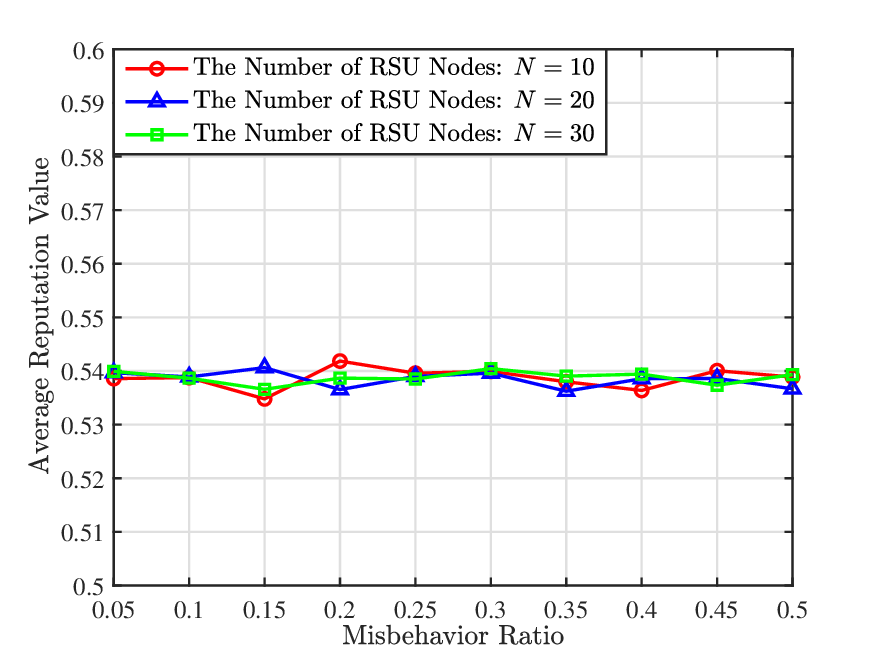}
\caption{The function of the average reputation value with respect to the misbehavior ratio.}
\label{mis_rep}
\end{figure}
\begin{figure}[t]
\centering
\includegraphics[width=0.4\textwidth]{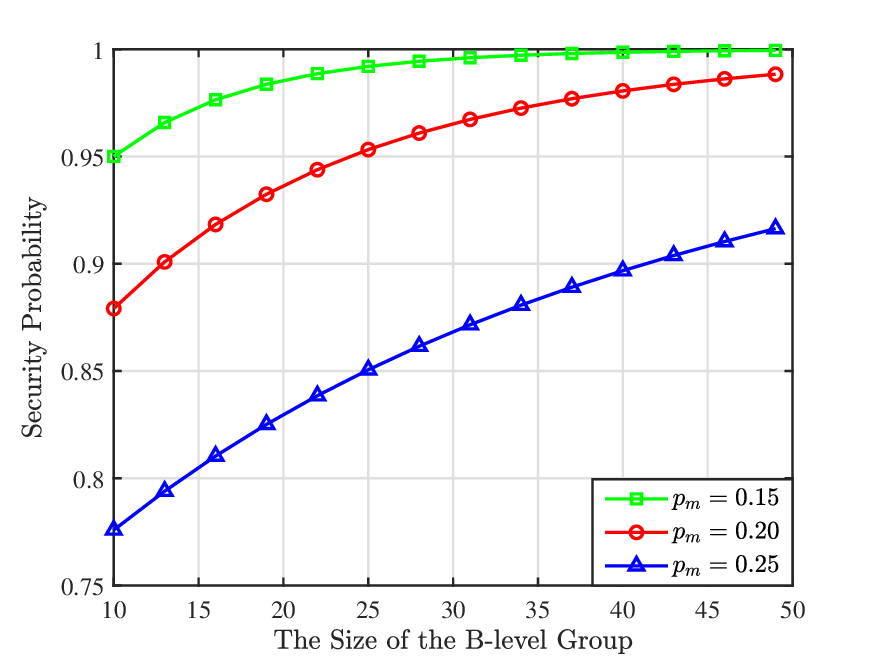}
\caption{Security probability under the different malicious probability of miners.}
\label{blockchain}
\end{figure}

\subsection{Performance of Secure Block Verification}
Figure \ref{blockchain} shows the security performance of the PBFT consensus algorithm in the proposed blockchain system. We consider that the A-level RSU group performs well and model the security performance of the B-level RSU group (i.e., a delegate group) as a random sampling problem with two possible results, namely malicious delegates and well-behaved delegates \cite{lei2020groupchain,qi2021privacy}. According to \cite{lei2020groupchain,li2020scalable}, when the number of malicious delegates is not higher than $(N-1)/3$, where $N$ is the total number of delegates, a block can be verified correctly and truly by the delegates. Therefore, the security probability of a delegate group is $P_{\text{safety}} = \sum_{z=0}^{\lfloor \frac{N}{3}\rfloor}\binom{N}{z}p_m^z(1-p_m)^{N-z}$, where $p_m$ is the probability of a delegate being malicious \cite{li2020scalable}. From Fig. \ref{blockchain}, we can see that regardless of the probability that malicious delegates exist, the security probability increases as the size of the B-level RSU group grows. The reason is that a larger size of the delegate group indicates the increased number of well-behaved delegates participating in the block verification, thereby ensuring security in the consensus process\cite{qi2021privacy}. Therefore, the proposed blockchain system with the PBFT consensus algorithm can ensure reliable VT migration by providing reliable and secure block verification.

\subsection{Numerical Analysis for the Stackelberg Game}

Figure \ref{strategy} shows the impacts of the user-centric parameter $\alpha$ on the optimal strategies of the single VMU and the RSU coalition. Figure ~\ref{vehicle} shows the optimal bandwidth demand strategy of the VMU under different $\alpha$ and $P$ when the cost of unit bandwidth $C=5$. Based on \eqref{U_v}, we can see that the amount of bandwidth purchased by the VMU is mainly affected by the user-centric parameter $\alpha$ and the selling price of unit bandwidth $P$ that the RSU coalition determines. With the increase of $\alpha$, the amount of bandwidth purchased by the VMU is increasing. The reason is that the larger $\alpha$ means that the VMU is more sensitive to the VT migration latency, so more bandwidth is requested to ensure the immersion of vehicular metaverse services. Besides, the higher the selling price of unit bandwidth $P$, the less amount of bandwidth that the VMU would purchase. For instance, when $\alpha=0.5$ and $P$ increase from $10$ to $30$, approximately $76\%$ reduction in bandwidth purchased by the VMU.

Figure ~\ref{RSU} shows the optimal pricing strategy of the RSU coalition under different $\alpha$ and $C$. The selling price of unit bandwidth determined by the RSU coalition is affected by the transmission cost of unit bandwidth $C$ and the user-centric parameter $\alpha$ based on \eqref{U_r}. The reason is that more bandwidth will be purchased by VMUs to ensure immersive metaverse experiences if $\alpha$ is larger. From Fig.~\ref{RSU}, we can see that no matter how much the cost of unit bandwidth $C$ is, the selling price of unit bandwidth $P$ increases with the increase of $\alpha$, and the higher the cost of unit bandwidth $C$, the higher the selling price of unit bandwidth $P$ to ensure the utility of the RSU coalition. For example, when the user-centric parameter $\alpha=0.5$ and the cost of unit bandwidth $C=5$, the RSU coalition sets the selling price of unit bandwidth $P$ at $19.6$ to incentivize VMUs to perform VT migration. When the cost of unit bandwidth $C=10$, the RSU coalition sets the selling price of unit bandwidth $P$ at $27.8$.

Figure \ref{utility} shows the impacts of $\alpha$ and $C$ on the utility of the RSU coalition. we can see that when the cost of unit bandwidth $C$ is fixed, the utility of the RSU coalition rises with the increase of the user-centric parameter $\alpha$, which is because the amount of bandwidth requested by VMUs increases. However, when $\alpha$ is held constant, the utility of the RSU coalition decreases as the cost of unit bandwidth $C$ increases, even though the selling price of unit bandwidth $P$ set by the RSU coalition also increases. This can be attributed to the impact of cost $C$ and the bandwidth demands of VMUs on the utility of the RSU coalition as defined by equation \eqref{U_r}. Thus, it is clear that while the increase in the cost of unit bandwidth $C$ leads to a higher selling price of unit bandwidth $P$, the reduction in bandwidth requests from VMUs has a greater negative effect on the utility of the RSU coalition.

\begin{figure}[t]
\centering
\subfigure[Impacts of $\alpha$ and $P$ on the optimal bandwidth demand strategy of the VMU.]{
\centering
\includegraphics[width=0.4\textwidth]{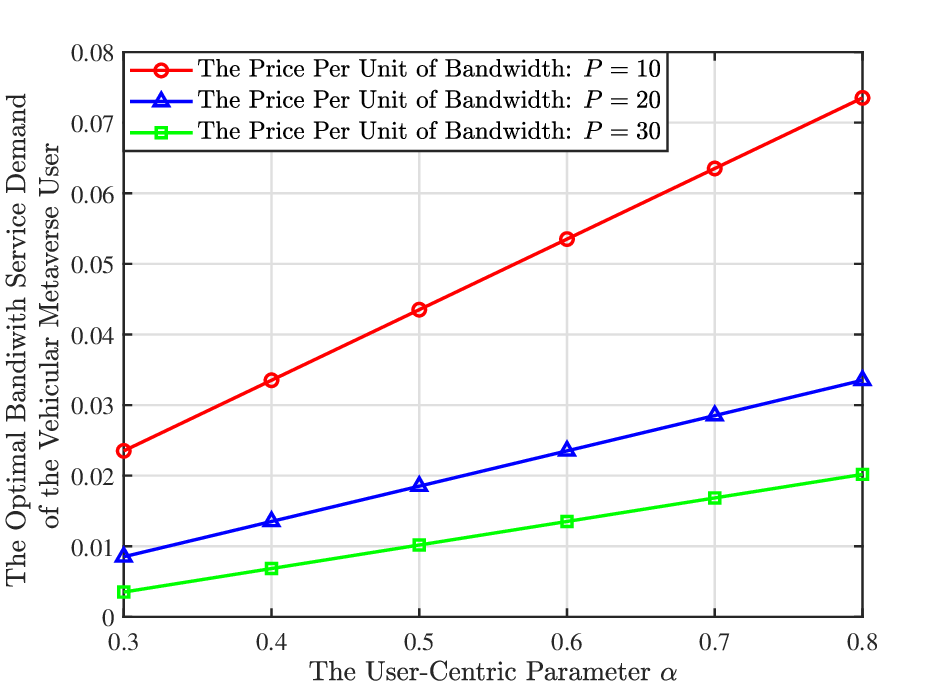}
\label{vehicle}
}
\subfigure[Impacts of $\alpha$ and $C$ on the optimal bandwidth pricing strategy of the RSU coalition.]{
\centering
\includegraphics[width=0.39\textwidth]{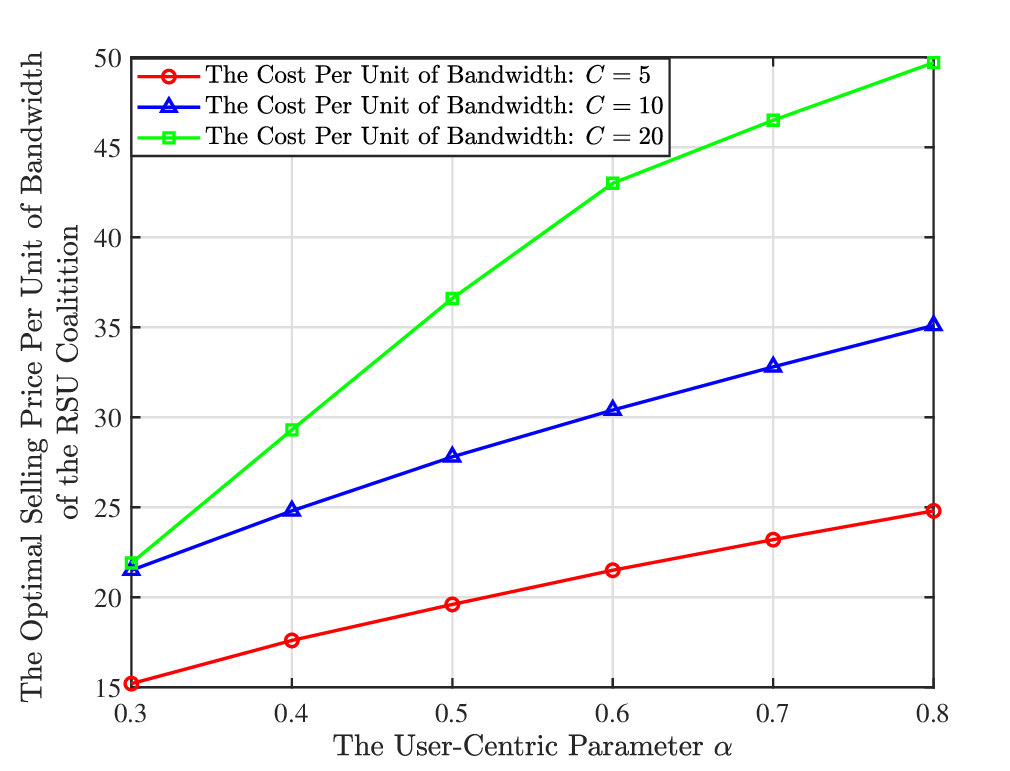}
\label{RSU}
}
\caption{Impacts of key parameters on the optimal strategies of the VMU and RSU coalition.}
\label{strategy}
\end{figure}

\begin{figure}[t]
\centering
\includegraphics[width=0.4\textwidth]{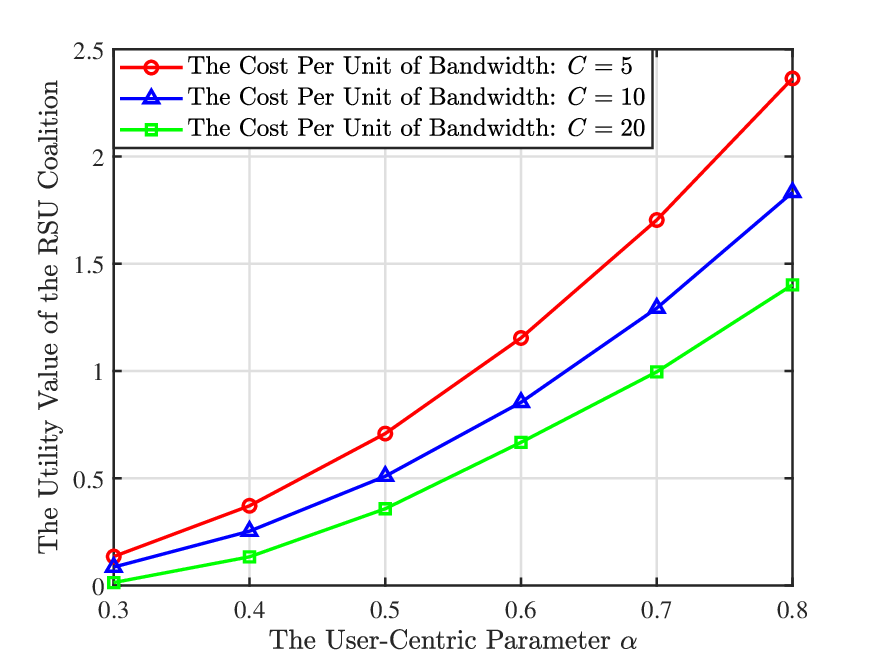}
\caption{Impacts of $\alpha$ and $C$ on the utility of the RSU coalition.}
\label{utility}
\end{figure}

\section{Conclusion and Future Work}\label{Conclusion}
In this paper, we proposed a blockchain-assisted game approach framework for VT migration in vehicular metaverses. To quantify the reliability of RSUs, we calculated the reputation values of RSUs based on the subjective logic model. Especially, we added a time attenuation factor by considering the interaction freshness. Besides, to manage reputation values securely, we proposed a blockchain-assisted reputation rating system, where RSUs as miners are divided into different levels according to their reputation values for lightweight consensus. Based on the reputation values of RSUs, we formulated a coalition game and formed RSU coalitions for reliable and large-scale VT migration. To incentivize VMUs for VT migration, we proposed a single-leader and multi-follower Stackelberg model between the RSU coalition with the highest utility and VMUs. Numerical results demonstrated the reliability and effectiveness of the proposed schemes.

The algorithm used to solve Stackelberg equilibrium can be further improved. Therefore, we will use AI technologies such as deep reinforcement learning to find the Stackelberg equilibrium in the future. Besides, we may construct a new immersion metric to optimize the utility of VMUs based on the characteristics of metaverse and link its computation to the field of psychology, e.g., by taking into account the enjoyment and engagement of VMUs with the services of the metaverse in the computation of immersion. Additionally, we recognize that the calculation formula for service latency can also be further improved, so we plan to enhance the calculation of service latency by utilizing communication-focused formulas for more accurate evaluations. This involves taking into account both the volume of data being transferred by VTs and the statistical distribution of data as it reaches the RSUs.

\section*{Acknowledgments}
The work was supported by National Natural Science Foundation of China (NSFC) under grant No. 62102099, U22A2054, and No. 62101594, and  the Pearl River Talent Recruitment Program under Grant 2021QN02S643, and Guangzhou Basic Research Program under Grant 2023A04J1699, and was also supported by National Natural Science Foundation of China (Grant No. 62071343), Foundation of State Key Laboratory of Public Big Data (No. PBD2023-12), and Collaborative Innovation Center of Novel Software Technology and Industrialization.

\bibliographystyle{IEEEtran}

\bibliography{ref}

\end{document}